\documentclass[11pt]{article}
\usepackage{graphicx}
\usepackage{algorithm,algorithmic}
\usepackage{amsmath, amssymb, amsthm}
\usepackage{url}
\usepackage{fullpage, prettyref}
\usepackage{pstricks,pst-node,pst-text}
\usepackage{boxedminipage}
\usepackage{enumerate}

\long\def\symbolfootnote[#1]#2{\begingroup%
\def\thefootnote{\fnsymbol{footnote}}\footnote[#1]{#2}\endgroup}

\newtheorem{theorem}{Theorem}
\newtheorem{lemma}{Lemma}

\newcommand{\ignore}[1]{}
\newcommand{\bs}{\backslash}

\newlength{\algobox}
\newenvironment{proofof}[1]{\smallskip\noindent{\bf Proof of #1.}}%
        {\hspace*{\fill}$\Box$\par}

\def\E{{\bf E}}
\def\lp{{\tt lp}}
\def\full{{\tt F}}

\def\Drop{{\tt Drop}}
\def\drop{{\tt drop}}
\def\Loss{{\tt Loss}}
\def\loss{{\tt loss}}
\def\MST{{\tt MST}}
\def\mst{{\tt mst}}
\def\CL{{\tt LC}}
\def\RL{{\tt RLC}}

\usepackage{xcolor}
\definecolor{foocite}{rgb}{0,0.75,0}
\definecolor{foolink}{rgb}{0,0,1}
\definecolor{foourl}{rgb}{1,0,0}
\usepackage[colorlinks=true,citecolor=foocite,linkcolor=foolink,urlcolor=foourl]{hyperref}

\begin{document}

\title{Integrality Gap of the Hypergraphic Relaxation of Steiner Trees: a short proof of a $1.55$ upper bound}
\author{ Deeparnab Chakrabarty \and Jochen K\"onemann \and  David Pritchard}
\date{\today}
\maketitle

\begin{abstract}
Recently, Byrka et al.~\cite{BGRS10} gave a $1.39$-approximation for
the Steiner tree problem, using a hypergraph-based LP relaxation. They
also upper-bounded its integrality gap by $1.55$. We describe a shorter
proof of the same integrality gap bound, by applying some of their techniques
to a randomized loss-contracting algorithm.
\end{abstract}

\section{Introduction}
In the Steiner tree problem, we are given an undirected graph
$G=(V,E)$ with costs $c$ on edges and its vertex set partitioned into
terminals (denoted $R \subset V$) and Steiner vertices ($V\setminus
R$). A \emph{Steiner tree} is a tree spanning all of $R$ plus any
subset of $V \bs R$, and the problem is to find a minimum-cost such
tree.  The Steiner tree problem is {$\mathsf{APX}$}-hard, thus the
best we can hope for is a constant-factor approximation algorithm.

The best known ratio is a result of Byrka, Grandoni, Rothvo{\ss} and
Sanit\`{a} \cite{BGRS10}: their randomized iterated rounding algorithm
gives approximation ratio $\ln(4)+\epsilon \approx 1.39$. The prior
best was a $1+\frac{\ln 3}{2}+\epsilon \approx 1.55$ ratio, via the
deterministic loss-contracting algorithm of Robins and
Zelikovsky~\cite{RZ05}. The algorithm of \cite{BGRS10}  differs from
previous work in that it uses a linear programming (LP)
relaxation; the LP is based on hypergraphs, and it has several
different-looking but equivalent~\cite{CKP10,PVd03} nice
formulations. A second result of~\cite{BGRS10} concerns the LP's
\emph{integrality gap}, which is defined as the worst-case ratio (max
over all instances) of the optimal Steiner tree cost to the LP's
optimal value. Byrka et al.~show the integrality gap is at most
$1.55$, and their proof builds on the analysis of
\cite{RZ05}. In this note we give a shorter proof of the same bound
using a simple LP-rounding algorithm.

%
%

\begin{figure}
\begin{center}
\includegraphics[scale=.75]{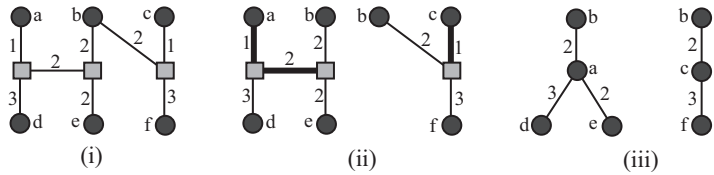}
\end{center}
\vspace{-0.5cm}
\caption{In (i) we show a Steiner tree; circles are terminals and squares are Steiner nodes. In (ii) we show its decomposition into full components, and their losses in bold. In (iii) we show the full components after loss contraction.}\label{fig:fc}
\end{figure}

We now describe one formulation for the hypergraphic LP. Given a set
$K \subset R$ of terminals, a {\em full component} on $K$ is a tree
whose leaf set is $K$ and whose internal nodes are Steiner vertices.
Every Steiner tree decomposes in a unique edge-disjoint way into full
components; Figure~\ref{fig:fc}(i) shows an example. Moreover, one
can show that a set of full components on sets $(K_1,\dotsc,K_r)$
forms a Steiner tree if and only if the hypergraph $(V,
(K_1,\dotsc,K_r))$ is a hyper-spanning tree. Let $\full(K)$ denote
a minimum-cost full component for terminal set $K \subset R$, and let $C_K$ be its cost.
The
hypergraphic LP is as follows:
\begin{align}
\min &\qquad \sum_{K} C_Kx_K: \label{eq:LP-hyp}
\tag{\ensuremath{\mathcal{S}}} \\
\forall \varnothing\neq S\subseteq R: &\qquad \sum_{K:K \cap S \neq \varnothing} x_K(|K \cap S|-1) \le |S|-1 \notag \\
 &\qquad \sum_{K} x_K(|K|-1) =  |R|-1 \notag \\
\forall K:  &\qquad  x_K\ge 0  \notag
\end{align}
The integral solutions of \eqref{eq:LP-hyp} correspond to the
full component sets of Steiner trees. As an aside, the \emph{$r$-restricted full component} method (e.g.~\cite{GH+01b}) allows us to assume there are a polynomial number of full components while affecting the optimal Steiner tree cost by a $1+\epsilon$ factor. Then, it is possible to solve \eqref{eq:LP-hyp} in polynomial time~\cite{BGRS10,War97}.
Here is our goal:
\ignore{
From Dave: I took this out since it no longer is clear and explicit like before in my opinion. If we're not going to explain in a little more detail, then non-experts will not be able to follow it, and we shouldn't say something possibly confusing.
  We note that \eqref{eq:LP-hyp} can be separated in polynomial time~\cite{PVd03,BGRS10} using
maybe solved (approximately) in polynomial time, by reducing the
number of variables, and by applying the Ellipsoid method. Reducing
the number of variables also allows us to assume that the full
components corresponding to the variables of the above LP are pairwise
edge- and Steiner vertex disjoint. We omit the details and refer
the reader to \cite{GH+01b} for details. }

\begin{theorem}\cite{BGRS10}
  The integrality gap of the hypergraphic LP \eqref{eq:LP-hyp} is at
  most $1 + \ln{3}/2 \approx 1.55$.
\end{theorem}\label{thm:main}

\section{Randomized Loss-Contracting Algorithm}

In this section we describe the algorithm. We introduce some
terminology first.  The \emph{loss} of full component $\full(K)$, denoted
by $\Loss(K)$, is a minimum-cost subset of $\full(K)$'s edges that connects the
Steiner vertices to the terminals. For example, Figure
\ref{fig:fc}(ii) shows the loss of the two full components in
bold. We let $\loss(K)$
denote the total cost of all edges in $\Loss(K)$.  The \emph{loss-contracted full
  component of $K$}, denoted by $\CL(K)$, is obtained from $\full(K)$ by contracting
its loss edges (see Figure \ref{fig:fc}(iii) for an example).

For clarity we make two observations. First, for each $K$ the edges of $\CL(K)$ correspond to the edges
of $\full(K) \bs \Loss(K)$. Second, for terminals $u, v$, there may be a $uv$ edge
in several $\CL(K)$'s but we think of them as distinct parallel edges. 

Our randomized rounding algorithm, \RL, is shown below. We
choose $M$ to have value at least $\sum_K x_K$ such that
$t = M\ln 3$ is integral. $\MST(\cdot)$ denotes a minimum
spanning tree and $\mst$ its cost.

\vspace{3ex}\noindent
\begin{boxedminipage}{\algobox}
{\bf Algorithm \RL.}
\begin{algorithmic}[1]
  \STATE Let $T_1$ be a minimum spanning tree of the induced graph $G[R]$.
  \STATE $x \leftarrow$ Solve \eqref{eq:LP-hyp} \label{rl:lp}
  \FOR{$1 \leq i \leq t$}
  \STATE Sample $K_i$ from the distribution\footnote{$K_i \leftarrow \varnothing$
    with probability $1-\sum_Kx_K/M$.} with probability $\frac{x_K}{M}$ for each full component $K$. \label{rl:sample}
  \STATE $T_{i+1} \leftarrow \MST(T_i \cup \CL(K_i))$ \label{rl:mst}
  \ENDFOR
  \STATE Output any Steiner tree in $ALG := T_{t+1} \cup \bigcup_{i=1}^t \Loss(K_i)$.
\end{algorithmic}
\end{boxedminipage}
\vspace{0.75ex}

To prove that $ALG$ actually contains a Steiner tree, we must show all terminals are connected. To see this, note each edge $uv$ of $T_{t+1}$ is either a terminal-terminal edge of $G[R]$ in the input instance, or else $uv \in \CL(K_i)$ for some $i$ and therefore a $u$-$v$ path is created when we add in $\Loss(K_i)$.

\section{Analysis}

In this section we prove that the
tree's cost is at most $1 + \frac{\ln 3}{2}$ times the optimum
value of \eqref{eq:LP-hyp}.
Each iteration of the main loop of algorithm \RL\ first samples a full
component $K_i$ in step \ref{rl:sample}, and subsequently recomputes a
minimum-cost spanning tree in the graph obtained from adding the
loss-contracted part of $K_i$ to $T_{i}$.  The new spanning tree
$T_{i+1}$ is no more expensive than $T_i$; some of its
edges are replaced by newly added edges in $\CL(K_i)$. Bounding the
drop in cost will be the centerpiece of our analysis, and this step
will in turn be facilitated by the elegant {\em Bridge Lemma} of
Byrka et al.~\cite{BGRS10}. We describe this lemma first.

\begin{figure}
  \begin{center}
    \includegraphics[scale=.8]{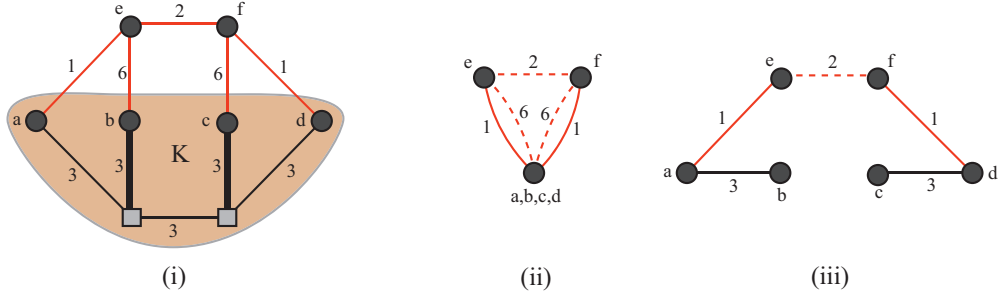}
  \end{center}
  \caption{\label{fig:drop} In (i) we show a terminal spanning tree
    $T$ in red, and a full component spanning terminal set $K \subset \{a,b,c,d\}$ in black; thick edges are its loss. In (ii) we show $T/K$, and $\Drop_T(K)$ is shown as dashed edges. In (iii) we show $\MST(T \cup \CL(K))$.}
\end{figure}

We first define the \emph{drop} of a full component $K$ with respect to a
terminal spanning tree $T$ (it is just a different name for the
bridges of~\cite{BGRS10}). Let $T/K$ be the graph obtained
from $T$ by identifying the terminals spanned by $K$.
Then let
$$ \Drop_T(K) := E(T)\setminus E(\MST(T/K)), $$
be the set of edges of $T$ that are not contained in a minimum
spanning tree of $T/K$, and $\drop_T(K)$ be its cost. We illustrate this in Figure
\ref{fig:drop}. We state the Bridge Lemma
here and present its proof for completeness.

\begin{lemma}[Bridge Lemma~\cite{BGRS10}]\label{lem:bridge}
Given a terminal spanning tree $T$ and a feasible solution $x$ to \eqref{eq:LP-hyp},
\begin{equation}
\sum_{K} x_K\drop_T(K) \ge c(T).\label{eq:fooo}
\end{equation}
\end{lemma}
\begin{proof}
    The proof needs the following theorem of Edmonds
    \cite{Ed71}: given a graph $H=(R,F)$, the extreme points of the
    polytope
    \begin{equation}
      \{z\in {\mathbb R}^F_{\ge 0}   :
      \sum_{(u,v)\in F: u\in S,v\in S}z_e \le |S|-1 \quad \forall S\subset R, \quad
            \sum_{e\in F} z_e = |R|-1\} \tag{$\mathcal{G}$} \label{graphic}
    \end{equation}
    are the indicator variables of spanning trees of $H$.
    The proof strategy is as follows. We construct a
    multigraph $H=(R,F)$ with costs $c$, and $z \in {\mathbb R}^F$ such that:
    the cost of $z$ equals
    the left-hand side of \eqref{eq:fooo}; $z \in \eqref{graphic}$;
    and all spanning trees of $H$ have cost at least $c(T)$. Edmonds' theorem then immediately
    implies the lemma. In the rest of the proof we define $H$ and supply the three parts of this strategy.


    For each full component $K$ with $x_K>0$, consider the
    edges in $\Drop_T(K)$.  Contracting all edges of $E(T)\setminus
    \Drop_T(K)$, we see that $\Drop_T(K)$ corresponds to edges of a
    spanning tree of $K$.  These edges are copied (with the same cost
    $c$) into the set $F$, and the copies are given weight $z_e =
    x_K$. Using the definition of drop, one can show each $e \in F$ is a
    maximum-cost edge in the unique cycle of $T \cup \{e\}$.
    
    Having now defined $F$, we see
    \begin{equation}\label{eq:xz}
      \sum_{e\in F}c_ez_e = \sum_K x_K\drop_T(K).
    \end{equation}
    Note that we introduce $|K|-1$ edges for each full component $K$,
    and that, for any $S \subseteq R$, at most $|S \cap K|-1$ of these
    have both ends in $S$. These two observations together with the
    fact that $x$ is feasible for \eqref{eq:LP-hyp} directly imply
    that $z$ is feasible for \eqref{graphic}.

    To show all spanning trees of $H$ have cost at least $c(T)$, it suffices to show $T$ is an MST of $T \cup H$.
    In turn, this follows (e.g.~\cite[Theorem 50.9]{Sc03}) from the fact that each $e \in F$ is a
    maximum-cost edge in the unique cycle of $T \cup \{e\}$.
  \end{proof}

We also need two standard facts that we summarize in the following
lemma. They rely on the input costs satisfying the triangle inequality, and that internal nodes
of full components have degree at least 3, both of which hold without loss of generality.

\begin{lemma} \label{facts}
  (a) The value $\mst(G[R])$ of the initial terminal spanning tree
    computed by algorithm \RL\ is at most twice the optimal value of
    \eqref{eq:LP-hyp}. (b) For any full component $K$, $\loss(K) \leq C_K/2$.
\end{lemma}
\begin{proof}See Lemma 4.1 in \cite{GH+01b} for a proof of (b). For (a) we use a shortcutting argument along with Edmonds' polytope \eqref{graphic} for the graph $H = G[R]$. In detail, let $x$ be an optimal solution to \eqref{eq:LP-hyp}. For each $K$, shortcut a tour of $\full(K)$ to obtain a spanning tree of $K$ with $c$-cost at most twice $C_K$ (by the triangle inequality) and add these edges to $F$ with $z$-value $x_K$. Like before, since $x$ is feasible for \eqref{eq:LP-hyp}, $z$ is feasible for \eqref{graphic}, and so there is a spanning tree of $G[R]$ whose $c$-cost is at most $\sum_{e \in F}c_e z_e \le 2\sum_K C_K x_K$.
\end{proof}

We are ready to prove the main theorem.

\begin{proofof}{Theorem \ref{thm:main}}
  Let $x$ be an optimal solution to \eqref{eq:LP-hyp} computed in step
  \ref{rl:lp}, define $\lp^*$ to be its objective value, and
  $$ \loss^* = \sum_K x_K\loss(K) $$
  its fractional loss.
  Our goal will be to derive upper bounds on the expected cost of tree
  $T_i$ maintained by the algorithm at the beginning of iteration $i$.
  \ignore{DP: I didn't understand why you phrased it this way... First consider the case where we deterministically pick full
  component $K_i$ in this iteration.}
  After selecting $K_i$, one possible candidate spanning
  tree of $T_i \cup \CL(K_i)$ is given by the edges of $T_i \setminus
  \Drop_{T_i}(K_i) \cup \CL(K_i)$, and thus
  \begin{equation}\label{eq:oldc}
    c(T_{i+1}) \le c(T_i) - \drop_{T_i}(K_i) + c(\CL(K_i)).
  \end{equation}
  \ignore{DP: I don't think this level of commentary is beneficial. Figure \ref{fig:drop}(iii) (with $T=T_i$ and $K=K_i$) illustrates
  the fact that the right-hand side of this inequality can sometimes
  be strictly larger than $c(T_{i+1})$; the terminal spanning tree
  given in the figure has cost $10$, while the bound on the right on
  right of \eqref{eq:oldc} is $11$.}

 Let us bound the expected value of $T_{i+1}$, given any
  fixed $T_i$. Due to the distribution from which $K_i$ is drawn, and
  using \eqref{eq:oldc} with linearity of expectation, we have
  $$E[c(T_{i+1})] \le c(T_i) - \frac{1}{M}\sum_K x_K\drop_{T_i}(K) + \frac{1}{M}\sum_K x_K(C_K-\loss(K)) .$$
  Applying the bridge lemma on the terminal spanning tree $T_i$, and
  using the definitions of $\lp^*$ and $\loss^*$, we have
  \begin{align}
    \E[c(T_{i+1})] 
    & \le (1 - \tfrac{1}{M})\E[c(T_i)] + (\lp^* - \loss^*)/M \notag
  \end{align}
  By induction this gives
  \begin{align}
    \E[c(T_{t+1})] &= (1 - \tfrac{1}{M})^t c(T_1) + (\lp^* - \loss^*)(1 - (1-\tfrac{1}{M})^t) \notag\\
    & \le \lp^*(1 + (1 - \tfrac{1}{M})^t) - \loss^*(1 - (1 - \tfrac{1}{M})^t). \notag
  \end{align}
  where the inequality uses Lemma \ref{facts}(a).
  The cost of the final Steiner tree is at most $c(ALG) \le c(T_{t+1}) + \sum_{i=1}^t \loss(K_i)$. Moreover,
  \begin{align}
    \E[c(ALG)] \le &~ \E[c(T_{t+1})] + t\cdot\loss^*/M  \notag \\
    \le & ~\lp^*(1 + (1 - \tfrac{1}{M})^t) + \loss^*((1 - \tfrac{1}{M})^t + \tfrac{t}{M} - 1) \notag\\
    \le& ~\lp^*\bigg(\frac{1}{2} + \frac{3}{2}\Big(1 - \frac{1}{M}\Big)^t +\frac{t}{2M}\bigg)
    \notag \\
    \mathop{\le} & ~\lp^*(1/2 + 3/2 \cdot \exp(-t/M) +t/2M)
    \notag
  \end{align}
  where the third inequality uses (a weighted average of) Lemma
  \ref{facts}(b).  The last line explains our choice of $t = M \ln 3$
  since $\lambda=\ln 3$ minimizes
  $\frac{1}{2}+\frac{3}{2}e^{-\lambda}+\frac{\lambda}{2}$, with value
  $1+\frac{\ln 3}{2}$. Thus the algorithm outputs a Steiner tree of
  expected cost at most $(1+\frac{\ln 3}{2})\lp^*$, which implies the
  claimed upper bound of $1+\frac{\ln 3}{2}$ on the integrality gap.
\end{proofof}

\medskip

We now discuss a variant of the result just proven.
A Steiner tree instance is \emph{quasi-bipartite} if there are no Steiner-Steiner edges. For quasibipartite instances, Robins and Zelikovsky tightened the analysis of their algorithm to show it has approximation ratio $\alpha$, where $\alpha \approx 1.28$ satisfies $\alpha = 1+\exp(-\alpha)$). Here, we'll show an integrality gap bound of $\alpha$ (the longer proof of~\cite{BGRS10} via the Robins-Zelikovsky algorithm can be similarly adapted). We can refine Lemma \ref{facts}(a) (like in~\cite{RZ05}) to show that in quasi-bipartite instances, $\mst(G[R]) \le 2(\lp^* - \loss^*)$. Continuing along the previous lines, we obtain
\begin{align*}
\E[c(ALG)]  \le  \lp^*(1+\exp(-t/M)) + \loss^*(t/M-1-\exp(-t/M))
\end{align*}
and setting $ t = \alpha M$ gives $\E[c(ALG)] \le \alpha \cdot \lp^*$, as needed. We note that in quasi-bipartite instances the hypergraphic relaxation is equivalent~\cite{CKP10} to the so-called \emph{bidirected cut relaxation} thus we get an $\alpha$ integrality gap bound there as well.

At the risk of numerology, we conclude by remarking that  $1+\frac{\ln3}{2}$  arose in two very different ways, by analyzing different algorithms (and similarly for $\alpha \approx 1.28$). A simple explanation for this phenomenon would be very interesting.

\bibliographystyle{plain}
\bibliography{biblio}

\end{document}